\documentclass{amsart}

\usepackage{tikz}
\usetikzlibrary{patterns}
\usepackage{caption}

\usepackage{amsmath, amsthm, amssymb, delarray, todonotes, graphicx, subcaption, float, url, hyperref} 
\usepackage[margin=1in]{geometry} 
\usepackage{enumerate}
\usepackage{tikz} 
\usetikzlibrary{patterns}
\usepackage{comment}

\usepackage[all]{xy}          
\xyoption{dvips}

\makeatletter
\renewcommand\l@subsection{\@tocline{2}{0pt}{2pc}{5pc}{}}
\makeatother

\newtheorem*{rep@theorem}{\rep@title}
\newcommand{\newreptheorem}[2]{%
	\newenvironment{rep#1}[1]{%
		\def\rep@title{#2 \ref{##1}}%
		\begin{rep@theorem}}%
		{\end{rep@theorem}}} 

\theoremstyle{plain}
\newtheorem{thm}{Theorem}[section]
\newreptheorem{thm}{Theorem}

\newreptheorem{thmp}{Theorem} 
\newtheorem{prop}[thm]{Proposition}

\theoremstyle{definition}
\newtheorem{defin}[thm]{Definition}
\newtheorem{example}[thm]{Example}
\newtheorem{def/ex}[thm]{Definition/Example}

\theoremstyle{remark}
\newtheorem{rem}[thm]{Remark}
\newtheorem{rems}[thm]{Remarks}

\newcommand{\refD}[1]{Definition~\ref{D:#1}}

\newcommand{\Z}{{\mathbb Z}}

\newcommand{\ac}{{\operatorname{ac}}}

\newcommand{\esc}{{\searrow\!\!\!\searrow}}

\begin{document}


\title[Hypergraphs and political structures]{Hypergraph models for political structures}


\author{Ismar Voli\'c}
\address{Department of Mathematics, Wellesley College, 106 Central Street, Wellesley, MA 02481}
\email{ivolic@wellesley.edu}
\urladdr{ivolic.wellesley.edu}

\author{Zixu Wang}
\address{Department of Mathematics, Wellesley College, 106 Central Street, Wellesley, MA 02481}
\email{zw10@wellesley.edu}



\begin{abstract}
Building on \cite{MV:Politics}, this paper extends the modeling of political structures from simplicial complexes to hypergraphs. This allows the analysis of more complex political dynamics where agents who are willing to form coalitions contain subsets that would not necessarily form coalitions themselves. We extend topological constructions such as wedge, cone, and collapse from simplicial complexes to hypergraphs and use them to study mergers, mediators, and power delegation in political structures. Concepts such as agent viability and system stability are generalized to the hypergraph context, alongside the introduction of the notion of local viability. Additionally, we use embedded homology of hypergraphs to analyze power concentration within political systems. Along the way, we introduce some new notions within the hypergraph framework that are of independent interest. 
\end{abstract}


\maketitle 

\tableofcontents


\parskip=6pt
\parindent=0cm


\section{Introduction}\label{S:Intro}

Hypergraphs are generalizations of both graphs and simplicial complexes in that they allow edges to connect an arbitrary number of vertices (like simplices in a simplicial complex do), but subsets of vertices that form hyperedges do not necessarily have to form hyperedges themselves (which is a requirement in simplicial complexes). This flexibility makes them a convenient tool for studying complex, non-binary relational systems, with applications in fields ranging from computational biology to network science \cite{AA, BBHMM, BFK, EERdDR, FGMR, FHZXC, KHZ, O-KMSR, Ramadan2004, XHXZYH, ZGYGLY} (a further list of applications can be found, for example, in \cite{AAINR:HypergraphPersistent}).

In this paper, hypergraphs model relationships among agents in a political system. Agents are represented by vertices, and if a subset of agents is willing to negotiate or enter a coalition, a hyperedge determined by those vertices is formed. In \cite{MV:Politics}, the same situation was modeled by a simplex in a simplicial complex, which imposed a somewhat unnatural condition that any subset of agents that form a coalition also must form a coalition. This, of course, is not the case in reality since some agents might only be willing to negotiate if other agent(s) are present. If such cohesive agents leave, the entire coalition might collapse. The case of a mediator brought in for coalition-building among feuding agents is an example of this. Hypergraphs do not impose a condition of closure under taking subsets, which allows for a more realistic model developed here.

Once the transition from simplicial complexes to hypergraphs is established, we devote the bulk of the paper to translating and carrying over the results from \cite{MV:Politics} to this new setting. This is sometimes straightforward and sometimes not since the topology of hypergraphs is not as established or well studied as that of simplicial complexes; hypergraphs are in the literature usually treated with extensions of graph-theoretic tools rather than topological ones. Our point of view is geometric, namely that hypergraphs are simplicial complexes embedded in a Euclidean space but with some faces missing, and this provides intuition for extending some basic topological notions to them: cone, suspension, join, wedge, and (strong) collapse, among others. As in \cite{MV:Politics}, these constructions are then translated into phenomena such as mergers, introduction of mediators, and delegation of power. 

Homology is another important topological notion for us, but the standard construction of completing a hypergraph to a simplicial complex and then taking its homology turns out to be insufficient. This is because the closure of a hypergraph adds faces/coalitions that are not there in the actual political structure. Luckily, a more suitable notion of \emph{embedded} homology for hypergraphs already exists \cite{BLRW} and this allows us to say something about the possible viabilities, the dynamics of relationships, and concentration of power in a political structure.

To assess the impact of different operations, we import the ideas of agent viability and system stability from \cite{MV:Politics} into the hypergraph framework. The results describing the interplay between these notions and various topological operations translate to hypergraphs in a straightforward way. We also introduce \emph{local} viability that captures the importance of an agent within their framework of coalitions. As we will argue, this idea turns out to carry richer information for hypergraphs than when restricted to simplicial complexes. 

Much like \cite{MV:Politics}, this paper only initiates the study of political structures via hypergraphs. Further avenues of investigation are suggested along the way, and the last section is also devoted to this. We feel that hypergraphs are a potentially fruitful approach to studying political structures with possible exciting extensions to other topics in game theory and social choice theory.






\subsection{Acknowledgments}
Ismar Voli\'c is grateful to the Simons Foundation for its support.

\section{Hypergraphs}\label{S:Hypergraphs}

This section introduces fundamental hypergraph concepts, with definitions that diverge slightly from those conventionally adopted in classical graph theory because we aim to establish a closer correspondence with simplicial complexes. This review is by no means exhaustive; these are many resources for further details, including, for example, \cite{B:(Hyper)graphs, B:HypergraphTheory}. Some definitions, like domination and elementary collapses, do not seem to appear in the literature, but are natural extensions of the same notions for simplicial complexes.


\subsection{Basic definitions}\label{S:BasicHypergraphs}

\begin{defin}\label{D:Hypergraph}
A \emph{hypergraph} \( H \) is a pair \( H = (V, E) \), where
\begin{itemize}
    \item \( V \) is a finite set whose elements are called \emph{vertices};
    \item \( E \) is a finite set whose elements are non-empty subsets of \( V \) called \emph{hyperedges}. The set of hyperedges must include all singleton subsets of \( V \), i.e., every vertex \( v \in V \) is also a hyperedge \( \{v\} \in E \).
\end{itemize}
\end{defin}

\begin{rems}\ 
\begin{enumerate}
\item In the usual definition of a hypergraph, vertices are not required to be in the hyperedge set. In our setup, vertices will correspond to agents and hyperedges to coalitions, and we want a single agent to be a valid coalition; hence the inclusion of elements of $V$ in $E$.
\item A \emph{simplicial complex} is a hypergraph with the additional condition that the edge set be ``downward closed,'' i.e.~closed under taking subsets. Thus if $e\subset V$ is a hyperedge (or a \emph{simplex} in the language of simplicial complexes), so are all subsets of $e$ (its \emph{faces}). For a review of simplicial complexes, see Appendix A of \cite{MV:Politics}.
\end{enumerate}
\end{rems}

Since the vertex set is contained in $E$, we will often refer to a hypergraph \( H \) just by its hyperedge set \( E \).
For a hypergraph \( H \) with \( k \) vertices, we label the vertices \( v_1, \ldots, v_k \), without prioritizing the order of the enumeration. 

\begin{example}\label{ex:hypergraph}
Let \( H = (V, E) \) be the hypergraph with the vertex set \( V = \{v_1, v_2, v_3, v_4, v_5\} \) and the hyperedge set
\begin{align*}
E = \{ &\{v_1\}, \{v_2\}, \{v_3\}, \{v_4\}, \{v_5\}, \{v_1, v_2\}, \{v_1, v_5\}, \\
       &\{v_2, v_4\}, \{v_3, v_4\}, \{v_1, v_3, v_4\}, \{v_1, v_2, v_3, v_4\} \}.
\end{align*}
This can be visualized in a standard way with vertices drawn as nodes and 
hyperedges depicted as closed curves containing the appropriate vertices. To simplify the visualization, we do not draw hyperedges that consist of a single vertex. Following this convention, the hypergraph in this example is represented by Figure \ref{fig:hypergraph-visual}.

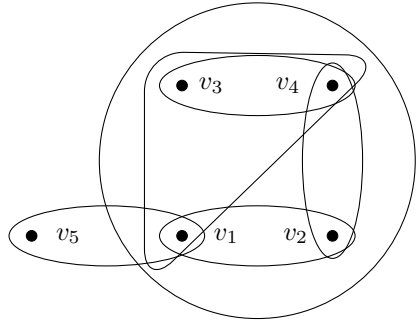
\begin{figure}[htbp]
\centering
\begin{tikzpicture}
    \tikzstyle{vertex} = [circle, draw, fill=black, inner sep=0pt, minimum width=4pt]

    \node[vertex] (v1) at (0, 0) {};
    \node[vertex] (v2) at (2, 0) {};
    \node[vertex] (v3) at (0, 2) {};
    \node[vertex] (v4) at (2, 2) {};
    \node[vertex] (v5) at (-2, 0) {};

    \draw (1, 0) ellipse (1.3cm and 0.4cm); 
    \draw (-1, 0) ellipse (1.3cm and 0.4cm); 
    \draw (2, 1) ellipse (0.4cm and 1.3cm); 
    \draw (1, 2) ellipse (1.3cm and 0.4cm); 
    \draw[rounded corners=15pt] (-0.5, -0.7) -- (-0.5, 2.45) -- (2.7, 2.4) -- cycle; 
    \draw (1, 1) circle (2.1cm); 

    \node[right=3mm] at (v1) {$v_1$};
    \node[left=2mm] at (v2) {$v_2$};
    \node[right=1mm] at (v3) {$v_3$};
    \node[left=3mm] at (v4) {$v_4$};
    \node[right=2mm] at (v5) {$v_5$};
\end{tikzpicture}
\caption{A visualization for the hypergraph of Example \ref{ex:hypergraph}.}
\label{fig:hypergraph-visual}
\end{figure}

\end{example}

\begin{defin}\label{def:VertexDegree}
The \emph{degree} of a vertex \( v \) in \( H \) is 
given by
\[
\deg(v) = |\{e \in E \,|\, v \in e\}|-1.
\]
\end{defin}
\begin{rem}
The usual way to define the degree of $v$ is as the number of hyperedges that contain it, but we subtract one because we do not want to include the hyperedge consisting of $v$ alone in the count. Isolated vertices, namely those that are not elements of any edges of size two of greater, thus have degree 0.
\end{rem}

\begin{example}\label{ex:VertexDegree}
For the hypergraph described in Example \ref{ex:hypergraph}, \( \text{deg}(v_4) = 4 \).
\end{example}

\begin{defin}\label{D:DimensionHypergraph}\ 
\begin{itemize}
\item The \emph{size} of a hyperedge \( e \) in a hypergraph \( H \)   is  \( |e| \), namely its cardinality. 
\item The \emph{dimension} of \( H \),  \( \dim(H) \), is the maximum size of a hyperedge, i.e.
\[ \dim(H) = \max_{e \in E} |e|. \].
\item The \emph{s-vector} of \( H \) is a vector \( (s_1, s_2, \ldots, s_{\dim(H)}) \), where \( s_i \) denotes the number of hyperedges of size \( i \). 
\item The \emph{s-vector of vertex \( v \)}, denoted by \( (s_1(v), s_2(v), \ldots, s_{m(v)}(v)) \), is a vector with  \( s_i(v) \) denoting the number of hyperedges of size \( i \) in \( H \) that contain the vertex \( v \) and \( m(v) \) is the size of the largest hyperedge containing $v$.
\end{itemize}
\end{defin}

\begin{example}\label{ex:DimensionHypergraph}\ 
\begin{itemize}
\item The dimension of the hypergraph in Example \ref{ex:hypergraph} is 4 because its largest hyperedge contains four vertices. Its $s$-vector is $(5, 4, 1, 1 )$. The $s$-vector of $v_1$ is $(1,2,1,1)$.
\item A 2-dimensional hypergraph is an ordinary graph.
\end{itemize}
\end{example}

It is immediate from the definitions that 
$\sum_{i=1}^{\dim(H)} s_i = |E|$ and $\sum_{i=1}^{m(v)} s_i(v) = \deg(v)+1.$

\begin{defin}\label{D:HyperedgeNeighborhood}\ 
The \emph{neighborhood of a hyperedge $e$} is the set of hyperedges in $H$ that have $e$ as a subset:
\[
\operatorname{N}(e) = \{h \in E \,|\, e \subset h\}.
\]
\end{defin}
Note that $e$ could be a vertex $v$, in which case the above defines a neighborhood $N(v)$ of $v$. It follows that $\deg(v) = |\operatorname{N}(v)| - 1$. In case the hypergraph is a simplicial complex, the notion of a neighborhood coincides with the \emph{star} of the hyperedge/face $e$, denoted by $\operatorname{st}(e)$. 
%


\begin{defin} Two vertices are \emph{adjacent}  if there exists a hyperedge containing both.
The \emph{adjacency count} of a vertex \( v \) is the number of vertices adjacent to $v$, namely
\[ \ac(v) = |\{ u \in V \setminus v\   | \text{ there exists } e \in E \text{ with } v , u\in e \}|. \]
\end{defin}

\begin{defin}\label{D:CompleteHypergraph} A hypergraph with the set of hyperedges $E = \mathcal{P}_0(V)$, namely the set of nonempty subsets of $V$, is called \emph{complete}. 
\end{defin}

In the terminology of simplicial complexes, if $|V|=k$, then the complete hypergraph on $V$ is a standard $(k-1)$-simplex $\Delta^{k-1}$. 

\begin{defin}\label{D:AssociatedSimplicialComplex}
Given a hypergraph $H$, its \emph{(simplicial) closure}, denoted by $K_H$, the smallest simplicial complex containing $H$.
\end{defin}

Thus $K_H$ is obtained by adding all subsets of elements already in $H$ to $H$. There is a natural inclusion map $H \to K_H$, providing a useful functor from the category of hypergraphs to the category of simplicial complexes.



\begin{defin}\label{D:MaximalHyperedge}
A \emph{maximal hyperedge} is a hyperedge that is not contained in another hyperedge. 
\end{defin}

\begin{example}
For the hypergraph in Example \ref{ex:hypergraph}, \(\{v_1, v_3, v_4\}\) is not a maximal hyperedge because it is a subset of the hyperedge \(\{v_1, v_2, v_3, v_4\}\), while the latter is a maximal hyperedge.
\end{example}

\begin{defin}\label{D:Domination}
Let \( v_i \) and \( v_j \) be vertices in a hypergraph. We say that \( v_i \) is \emph{dominated} by \( v_j \), denoted as \( v_i \preceq v_j \), if the following conditions hold:
\begin{enumerate}
  \item Every maximal hyperedge that contains \( v_i \) also contains \( v_j \).
  \item There exists a hyperedge \(\{v_i, v_j\}\).
\end{enumerate}
A hypergraph without any dominated vertices is called \emph{minimal}.

\begin{rem}
In the realm of simplicial complexes, domination is defined using only the first condition. However, for our purposes it also makes sense to impose the second one. Namely, if an agent is dominated by another agent in a political structure, we will want to consider the situation when the former relinquishes their power to the latter; this is modeled by the elementary strong collapse defined below. But for this, we want to assume that the two agents are willing to talk to each other, namely be in a coalition of their own, and this is represented by an edge containing them. 
\end{rem}
\end{defin}

\begin{example}
For the hypergraph in Example \ref{ex:hypergraph}, $v_3$ is dominated by $v_4$, and $v_5$ is dominated by $v_1$. However, $v_1$ is not dominated by $v_4$ because there does not exist a hyperedge that contains precisely $v_1$ and $v_4$.
\end{example}

\begin{defin}\label{D:HypergraphDeletion}
The \emph{deletion of a vertex $v$} in a hypergraph $H$, denoted by $H \setminus v$, is the subhypergraph of $H$ obtained by removing the neighborhood of $v$ from it, i.e.~retaining all hyperedges in $H$ that do not contain $v$:
\[
H \setminus v = \{e \in E \,|\, v \notin e\}.
\]
\end{defin}

\begin{defin}\label{D:ElemStrongCollapse}\ 
\begin{itemize}
\item The deletion of a dominated vertex $v$ from a hypergraph $H$ is called an \emph{elementary strong collapse}, denoted by $H\, \esc\, H\setminus v$. 
%
\item A sequence of elementary strong collapses is a \emph{strong collapse}.
\item The inverse of a strong collapse is a \emph{strong expansion}. 
\item If $\operatorname H_1$ and $H_2$ are related by a sequence of strong collapses and expansions, then they have the same \emph{strong homotopy type}.
\item A hypergraph $H$ is \emph{strongly collapsible} if it has the strong homotopy type of a hypergrph consisting of a single vertex.
\end{itemize}
\end{defin}

\begin{example}\label{Ex:StrongCollapse}
Figure \ref{fig:collapse} below gives an example of a strong collapse.

\end{example}

In a  minimal hypergraph, no elementary strong collapses are possible. Each hypergraphs admits a unique minimal subhypergraph via a strong collapse called the  \emph{core of $H$} and denoted by $H^c$. For example, picture (E) of Figure \ref{fig:collapse} (single vertex) is the core of all the hypergraphs before it.

\begin{defin}\label{D:HypeMap}
A \emph{hypergraph map} $\phi\colon \operatorname H_1\to H_2$ between hypergraphs  $\operatorname H_1$ to $H_2$ is a function that sends vertices to vertices and hyperedges to hyperedges.  A bijective hypergraph map whose inverse is also a hypergraph map is an \emph{isomorphism}.
\end{defin}


\subsection{Geometric realization}\label{S:Realization}


The standard way to visualize and topologize a simplicial complex $K$ is to associate to it its geometric realization $|K|$. Each simplex of $K$ of cardinality $n+1$ corresponds to the standard geometric $n$-simplex in Euclidean space (and topologized as its subset), and the simplices are glued along common faces. The same procedure can be performed for hypergraphs. One way to think about the realization $|H|$ of $H$ is to first  take the geometric realization $|K_H|$ of the simplicial closure of $H$, and then remove faces corresponding to elements of $K_H\setminus H$.

\begin{example}The geometric realization of the hypergraph in Example \ref{ex:hypergraph} is given in picture (A) of Figure \ref{fig:collapse}. Dashed lines indicate that the corresponding 2-vertex hyperedges do not exist. The filled-in triangle indicates that the hypergraph contains the hyperedge \(\{v_1, v_3, v_4\}\). Uncolored triangles indicate that the 3-vertex hyperedges corresponding to those faces do not exist. Additionally, there is a  tetrahedron determined by vertices $v_1, v_2, v_3, v_4$ because there is a hyperedge \(\{v_1, v_2, v_3, v_4\}\). We refrained from trying to color in the tetrahedron in order to make the picture simpler.

The rest of Figure \ref{fig:collapse} gives an illustration of a strong collapse. Since the last pictures is that of a single vertex, the hypergraph is strongly collapsible.

\begin{figure}[htbp]
\centering
\begin{subfigure}{0.45\textwidth}
        \centering
        \begin{tikzpicture}
           \draw[dashed] (-1.6,-0.1) -- (0,1.5);
\draw (0,1.5) -- (1.2,0.4);
\draw[dashed] (0,-1) -- (1.2,0.4);

\draw[dashed] (-1.6,-0.1) -- (0,1.5);
\draw[dashed] (0,-1) -- (1.2,0.4);

\draw[dashed] (1.1,0.4) -- (-1.5,-0.1) ;
\draw (0,-1) -- (0,1.5); 
\draw (-1.6,-0.1) -- (0,-1); 
\draw (-1.6,-0.1) -- (-3.2,-0.5); 

\fill[fill=blue!30, opacity=0.5] (-1.6,-0.1) -- (0,1.5) -- (0,-1) -- cycle;
\fill[fill=blue!30, opacity=0.5] (-1.6,-0.1) -- (0,1.5) -- (1.2,0.4) -- cycle;
\fill[fill=blue!30, opacity=0.5] (-1.6,-0.1) -- (0,-1) -- (1.2,0.4) -- cycle;

\node[above] at (-1.6,-0.1) {\scriptsize $v_1$};
\node[above] at (0,1.5) {\scriptsize $v_4$};
\node[above] at (1.2,0.4) {\scriptsize $v_3$};
\node[below] at (0,-1) {\scriptsize $v_2$};
\node[above] at (-3.2,-0.4) {\scriptsize $v_5$};

\foreach \coord in {(-1.6,-0.1), (0,1.5), (1.2,0.4), (0,-1), (-3.2,-0.5)}
    \filldraw[black] \coord circle (1pt);
    
        \end{tikzpicture}
        \caption{}
        \label{fig:sub1}
    \end{subfigure}\hfill
    \begin{subfigure}{0.45\textwidth}
        \centering
        \begin{tikzpicture}
           \draw[dashed] (-1.6,-0.1) -- (0,1.5);
\draw (0,1.5) -- (1.2,0.4);
\draw[dashed] (0,-1) -- (1.2,0.4);

\draw[dashed] (-1.6,-0.1) -- (0,1.5);
\draw[dashed] (0,-1) -- (1.2,0.4);

\draw[dashed] (1.1,0.4) -- (-1.5,-0.1) ;
\draw (0,-1) -- (0,1.5); 
\draw (-1.6,-0.1) -- (0,-1); 

\fill[fill=blue!30, opacity=0.5] (-1.6,-0.1) -- (0,1.5) -- (0,-1) -- cycle;

\node[above] at (-1.6,-0.1) {\scriptsize $v_1$};
\node[above] at (0,1.5) {\scriptsize $v_4$};
\node[above] at (1.2,0.4) {\scriptsize $v_3$};
\node[below] at (0,-1) {\scriptsize $v_2$};

\foreach \coord in {(-1.6,-0.1), (0,1.5), (1.2,0.4), (0,-1)}
    \filldraw[black] \coord circle (1pt);
    \fill[fill=blue!30, opacity=0.5] (-1.6,-0.1) -- (0,1.5) -- (1.2,0.4) -- cycle;
\fill[fill=blue!30, opacity=0.5] (-1.6,-0.1) -- (0,-1) -- (1.2,0.4) -- cycle;
    
        \end{tikzpicture}
        \caption{}
        \label{fig:sub2}
    \end{subfigure}

    \medskip

    \begin{subfigure}{0.3\textwidth}
        \centering
        \begin{tikzpicture}
           \draw[dashed] (-1.6,-0.1) -- (0,1.5);

\draw[dashed] (-1.6,-0.1) -- (0,1.5);

\draw (0,-1) -- (0,1.5); 
\draw (-1.6,-0.1) -- (0,-1); 

\fill[fill=blue!30, opacity=0.5] (-1.6,-0.1) -- (0,1.5) -- (0,-1) -- cycle;

\node[above] at (-1.6,-0.1) {\scriptsize $v_1$};
\node[above] at (0,1.5) {\scriptsize $v_4$};
\node[below] at (0,-1) {\scriptsize $v_2$};

\foreach \coord in {(-1.6,-0.1), (0,1.5), (0,-1)}
    \filldraw[black] \coord circle (1pt);

        \end{tikzpicture}
        \caption{}
        \label{fig:sub3}
    \end{subfigure}\hfill
    \begin{subfigure}{0.3\textwidth}
        \centering
        \begin{tikzpicture}
          \draw (-1.6,-0.1) -- (0,-1);

\node[above] at (-1.6,-0.1) {\scriptsize $v_1$};
\node[above] at (0,-1) {\scriptsize $v_2$};

\foreach \coord in {(-1.6,-0.1), (0,-1)}
    \filldraw[black] \coord circle (1pt);

        \end{tikzpicture}
        \caption{}
        \label{fig:sub4}
    \end{subfigure}\hfill
    \begin{subfigure}{0.3\textwidth}
        \centering
        \begin{tikzpicture}
          \node[above] at (-1.6,-0.1) {\scriptsize $v_1$};

\foreach \coord in {(-1.6,-0.1)}
    \filldraw[black] \coord circle (1pt);

        \end{tikzpicture}
        \caption{}
        \label{fig:sub5}
    \end{subfigure}
\caption{Illustration of a strong collapse on a sequence of geometric realizations.}
\label{fig:collapse}
\end{figure}
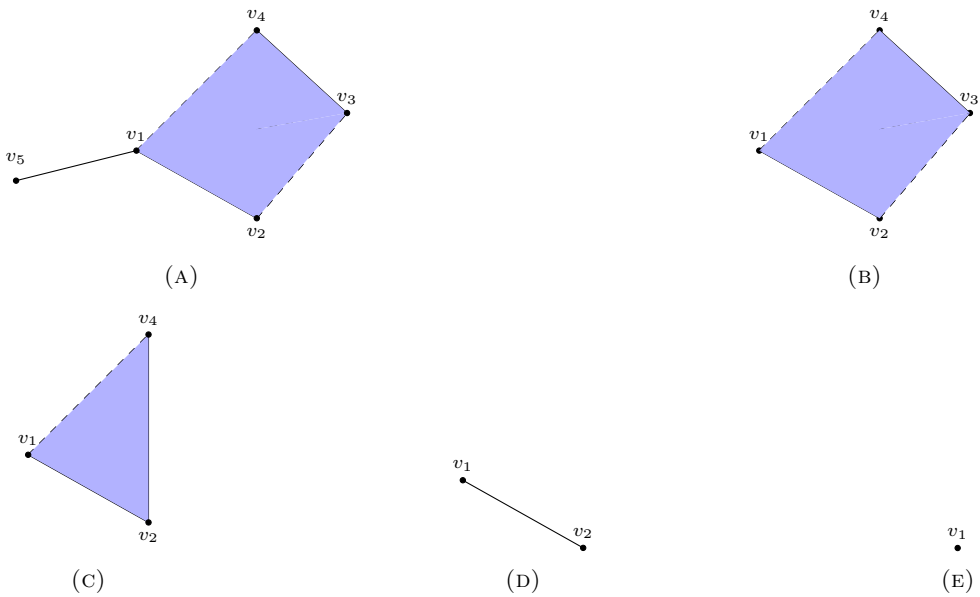

%
%
%
%
%
%
%
\end{example}

From now on, we will use $H$ for a hypergraph $H$ and its realization $|H|$ and will not distinguish between the two.



\subsection{Operations on hypergraphs}\label{S:HypergraphOperations}


There are several standard constructions that can be performed on simplicial complexes, and they have straightforward generalizations to hypergraphs.

\begin{defin}\label{D:Wedge}
Let $\operatorname H_1$ and $H_2$ be hypergraphs with distinguished vertices $v$ and $w$.  Define the \emph{wedge} of $\operatorname H_1$ and $H_2$, denoted by $\operatorname H_1\vee H_2$, to be the the union of $\operatorname H_1$ and $H_2$, except that one of $v$ or $w$ is removed (and the remaining one is relabeled if necessary).
\end{defin}


Geometrically, $\operatorname H_1\vee H_2$ is obtained as a quotient space
$$
\operatorname H_1\vee H_2 = (\operatorname H_1\sqcup H_2)/(v\sim w),
$$
where $\sqcup$ denotes disjoint union. Thus $\operatorname H_1\vee H_2$ is simply the space obtained by identifying $v$ and $w$, i.e.~``attaching'' $\operatorname H_1$ and $H_2$ along these points.

One can also attach hypergraphs along common subhypergraphs in a more general construction called the \emph{pushout}. 

Another general construction is the \emph{join}, obtained by taking the union of $\operatorname H_1$ and $H_2$ as well as the unions of the elemens of all the possible pairs of hyperedges from $\operatorname H_1$ and $H_2$. A special case is

%

\begin{defin}\label{D:Cone} Suppose $\{c\}$ is a hypergraph consisting of a single vertex. The \emph{cone} on a hypergraph $H$ is the hypegraph $CH$ given by 
$$
CH= H\sqcup \{c\}\sqcup \{e\cup c\colon e\in H\}.
$$
Vertex $c$ is the \emph{cone vertex} or the \emph{cone point}.  
\end{defin}

In terms of realizations, the cone is the quotient $(H\times I)/(H\times \{1\})$, where $I=[0,1]$.

\begin{example}
Figure \ref{fig:cone} gives an illustration of the cone construction.  All three vertices from \( H \) are connected to the cone point \( c \) by the edges \(\{ v_1, c\}\), \(\{ v_2,c \}\), and \(\{ v_3,c \}\). The edges \(\{ v_1, v_2 \}\) and \(\{ v_1,v_3 \}\) of \( H \) form hyperedges \( \{ v_1, v_2, c \}\) and \(\{ v_1,v_3,c \}\) in \( CH \). Since the edge \(\{ v_2,v_3\} \) is missing in \( H \),  the face \(\{ v_2, v_3, c \}\) is also absent. The hyperedge \(\{v_1, v_2, v_3\}\) of \( H \) leads to the hyperedge \(\{v_1, v_2, v_3, v_4\}\) in \( CH \), and as a result, the tetrahedron is solid.

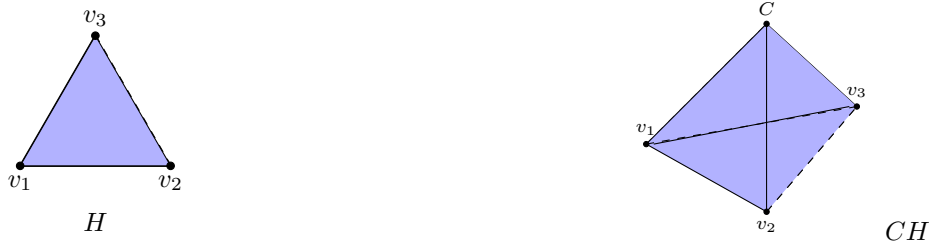
\begin{figure}[htbp]
  \begin{minipage}{0.45\textwidth}
    \centering
    \begin{tikzpicture}
      \coordinate (v1) at (0,0);
      \coordinate (v2) at (2,0);
      \coordinate (v3) at (1,1.732); 
      
      \draw[thick] (v1) -- (v2);
      \draw[thick] (v1) -- (v3);
      
      \draw[thick, dashed] (v2) -- (v3);
  
      \filldraw[fill=blue!30, opacity=0.5] (v1) -- (v2) -- (v3) -- cycle;
  
      \foreach \v in {v1, v2, v3} {
          \draw[fill=black] (\v) circle (0.05);
      }
  
      \node at (v1) [below] {$v_1$};
      \node at (v2) [below] {$v_2$};
      \node at (v3) [above] {$v_3$};
    \end{tikzpicture}\\
    \( H \)
  \end{minipage}
  \hfill
  \begin{minipage}{0.45\textwidth}
    \centering
    \begin{tikzpicture}

\draw[dashed] (-1.6,-0.1) -- (0,1.5);
\draw (0,1.5) -- (1.2,0.4);
\draw[dashed] (0,-1) -- (1.2,0.4);

\fill[fill=blue!30, opacity=0.5] (-1.6,-0.1) -- (0,1.5) -- (0,-1) -- cycle;
\fill[fill=blue!30, opacity=0.5] (-1.6,-0.1) -- (0,1.5) -- (1.2,0.4) -- cycle;
\fill[fill=blue!30, opacity=0.5] (-1.6,-0.1) -- (0,-1) -- (1.2,0.4) -- cycle;

\draw (-1.6,-0.1) -- (0,1.5);
\draw[dashed] (0,-1) -- (1.2,0.4);
\draw (1.1,0.4) -- (-1.5,-0.1) ;
\draw (-1.6,-0.1) -- (0,-1); 
\draw (0,-1) -- (0,1.5); 
\draw[dashed] (-1.6,-0.1) -- (1.2,0.4);

\node[above] at (-1.6,-0.1) {\scriptsize $v_1$};
\node[above] at (0,1.5) {\scriptsize $C$};
\node[above] at (1.2,0.4) {\scriptsize $v_3$};
\node[below] at (0,-1) {\scriptsize $v_2$};

\foreach \coord in {(-1.6,-0.1), (0,1.5), (1.2,0.4), (0,-1)}
    \filldraw[black] \coord circle (1pt);

\end{tikzpicture}
    \( CH \)
  \end{minipage}
  \caption{Illustration of the cone operation.}
  \label{fig:cone}
\end{figure}
\end{example}

\subsection{Hypergraph homology}\label{S:Homology}


The standard way to arrive at the homology of a hypergraph $H$ is to regard it as a poset (ordered by inclusion) and then construct the nerve of that poset. The homology of the resulting space is the homology of the $H$. It turns out that this homology is isomorphic to the standard homology of the simplicial complex $K_H$ obtained by the closure of the hypergraph.

However, we find this definition insufficient for our purposes because we wish for homology to ``remember'' which faces are missing. Instead, a more relevant notion of the homology of hypergraphs is that of \emph{embedded homology} \cite{BLRW} (which generalizes \emph{path homology} \cite{GLMY:PathHomology}). In brief, the chain complex that computes the embedded homology of $H$ is the smallest subcomplex of the singular chain complex for $K_H$ that contains the graded group generated by the hyperedges of $H$. In case $H$ is a simplicial complex, this produces the usual homology. But $H$ and $K_H$ can in general have different homology groups, as illustrated in the example below. We will denote the $n$th embedded homology o a hypergraph $H$ by $\operatorname H^{emb}_n(H)$.

\begin{example} Figure \ref{fig:closure} shows a hypergraph $H$ and its closure $K_H$. It turns out that
$\operatorname H^{emb}_0(H) = \mathbb{Z} \oplus \mathbb{Z}$ while  $\operatorname H_0(K_H) = \mathbb{Z}$ (see \cite{BLRW} for details). Adding one of the two missing edges to $H$ would change its 0th homology to $\Z$ (as calculated in \cite[Example 3.7]{GWXW}). Intuitively, this makes sense since contracting the single edge in $H$ would produce two vertices, i.e.~two componenents, while both $H$ with an added edge and $K_H$ are contractible and hence consist of one component.

\begin{figure}[htbp]
  \begin{minipage}{0.45\textwidth}
    \centering
    \begin{tikzpicture}
      \coordinate (v1) at (0,0);
      \coordinate (v2) at (2,0);
      \coordinate (v3) at (1,1.732); 
      
      \draw[thick] (v1) -- (v2);
      
      \draw[thick, dashed] (v1) -- (v3);
      \draw[thick, dashed] (v2) -- (v3);
  
      \filldraw[fill=blue!30, opacity=0.5] (v1) -- (v2) -- (v3) -- cycle;
  
      \foreach \v in {v1, v2, v3} {
          \draw[fill=black] (\v) circle (0.05);
      }
  
      \node at (v1) [below] {$v_1$};
      \node at (v2) [below] {$v_2$};
      \node at (v3) [above] {$v_3$};
    \end{tikzpicture}\\
    \( H \)
  \end{minipage}
  \hfill
  \begin{minipage}{0.45\textwidth}
    \centering
    \begin{tikzpicture}
    \coordinate (v1) at (0,0);
      \coordinate (v2) at (2,0);
      \coordinate (v3) at (1,1.732); 
      
      \draw[thick] (v1) -- (v2);
      \draw[thick] (v1) -- (v3);
      \draw[thick] (v2) -- (v3);
  
      \filldraw[fill=blue!30, opacity=0.5] (v1) -- (v2) -- (v3) -- cycle;
  
      \foreach \v in {v1, v2, v3} {
          \draw[fill=black] (\v) circle (0.05);
      }
  
      \node at (v1) [below] {$v_1$};
      \node at (v2) [below] {$v_2$};
      \node at (v3) [above] {$v_3$};
    \end{tikzpicture}\\
    \( K_H \)
  \end{minipage}
  \caption{}
\label{fig:closure}
\end{figure}
\end{example}


\section{Modeling political structures with hypergraphs}\label{S:Model}


The premise of \cite{MV:Politics} is that, given a collection of \emph{agents} in a political system (voters, political parties, members of a board of directors, etc.), we can regard those as vertices of a simplicial complex and simplices as potential coalitions. The presence of a simplex indicates that the agents represented by its vertices are willing to negotiate or vote the same way on a resolution. Various topological ideas and constructions on simplicial complexes then have interpretations and consequences for the functioning of the political system.

In this section, we extend much of what was done in \cite{MV:Politics} to hypergraphs. One restrictive assumption in that paper is that any subset of agents in a viable configuration also forms a viable configuration. This had to be true as simplicial complexes require closure under subsets. Current extension to hypergraphs eliminates this assumption, as it should since there is no guarantee a coalition will stay intact after one or more agents leave it. If an agent departs a coalition, other agents might do so as well or the entire coalition might fall apart. Dropping the requirement that every subset of a hyperedge is also a hyperedge allows for this possibility.

%
%


\subsection{Political structures}\label{S:Dictionary}


This section is a straightforward extension of the basic dictionary set up in \cite{MV:Politics}. We follow the terminology from that paper.

%
%
%
%

We make the following definition after \cite[Definition 1]{AK:Conflict} and \cite[Definition 3.1]{MV:Politics}.

\begin{defin}\label{D:PoliticalStructure}
A \emph{political structure} \( P \) is an ordered pair \( P = (A, \mathcal{C}) \), where
\begin{itemize}
    \item \( A \) is a finite set whose elements are called \emph{agents};
    \item \( \mathcal{C} \) is a collection of non-empty subsets of \( A \) called \emph{viable configurations} which includes all elements of \( A \), i.e., every agent \( a \in A \) forms a viable configuration \( \{a\} \) on their own.
\end{itemize}
\end{defin}

A correspondence between political structures and hypergraphs should be clear; vertices are agents and hyperedges are viable configurations, which should be thought of a coalitions. From now on, we will blur the disctinction between the two and will use the terms political structure and hypergraph interchangeably.

A \emph{map} $P\to Q$ between political structures is simply a hypergraph map of underlying hypergraphs. Such a map might capture a change, represent a transformation of the system after an event, a consolidation of agents, or an embedding of one structure into another. In particular, agents are sent to agents, but not necessarily in an injective of surjective way, meaning that multiple agents might have been consolidated or merged, or that new agents have been introduced. Since hyperedges map to hyperedges, viable configurations are mapped to viable configuration, indicating that, if agents were compatible before, they are still compatible after the event modelled by the map.

If there are $k$ agents and all configurations are viable, then we get the complete hypergraph or a $(k-1)$-simplex $\Delta^{k-1}$. We say the structure is in this case \emph{fully viable}. If no coalitions are possible, we get a 0-dimensional hypergraph consisting only of $k$ vertices and no other hyperedges. If all viable configurations are of size 2, the structure is modeled by an ordinary graph.

%
%
%
%
%

We can interpret various other definitions from Section \ref{S:Hypergraphs} in the context of political structures.

\begin{defin}\label{D:PoliticalSystemStuff}\ 

\begin{itemize}
\setlength\itemsep{4pt}
\item The \emph{dimension} of a political structure is the size of its largest viable configuration. 
\item A viable configuration is \emph{maximal} if it is not contained in any other configuration. 
\item An agent $a_j$ is \emph{more central} than agent $a_i$ if $a_i$ is dominated by $a_j$ in the corresponding hypergraph. 
\end{itemize}

\end{defin}

%


\subsection{(Local) viability and stability}\label{S:Stability}


One of the main goals of modeling political structures by hypergraphs (or simplicial complexes) is that their topology, geometry, and combinatorics  allow us to define and quantify how stable the system is and how well-positioned its agents are individually or in coalitions. This can then in turn be used for deciding how the structure can be further stabilized or where the most urgent interventions are.

To make this precise, we provide the following three definitions. The notion of local viability is new, while the other two appear in \cite[Definition 3.8]{MV:Politics} and are modified appropriately to hypergraphs here.

%


\begin{defin}\label{D:Viability}
Given a political structure $P=(A,\mathcal C)$ with the set of agents $A=\{a_1, ..., a_k\}$, $k\geq 2$, and the set of hyperedges $\mathcal C$, define the 
\begin{itemize}

\item \emph{viability of $a_i$} to be
\begin{equation*}\label{E:Viability}
\operatorname{via}(a_i)=\frac{|\operatorname{N}(a_i)|-1}{2^{k-1}-1}; 
\end{equation*}

\item 
\emph{local viability of $a_i$}, where $a_i$ is not an isolated agent (i.e.~$|\operatorname{N}(a_i)|>1$) to be
\begin{equation*}\label{E:LocalViability}
\operatorname{lvia}(a_i) = \frac{|\operatorname{N}(a_i)|-1}{2^{\ac(a_i)}-1};
\end{equation*}
In case $a_i$ is isolated, define $\operatorname{lvia}(a_i) = 0$.
\item \emph{stability of $P$} to be
\begin{equation*}\label{E:Stability}
\operatorname{stab}(P)=\frac{|\mathcal C|-k}{2^{k}-k-1}. 
\end{equation*}
\end{itemize}
\end{defin}

The viability of $a_i$ is determined by the number of hyperedges that contain $a_i$ as a vertex,  normalized to range between 0 and 1. A higher value of $\operatorname{via}(a_i)$ means that $a_i$ has many coalition options and is well-connected or well-trusted in the system. An isolated agent has viability 0. On the other hand, an agent who forms every possible coalition with other agents has viability 1. The minimal hypergraph for which this is possible is the the $(k-1)$-simplex $\Delta^{k-1}$ with the face $\{a_1, ..., \widehat{a_i}, ..., a_k\}$ removed (where $\widehat{a_i}$ indicates the omission of $a_i$). Thus if $P$ is represented by  $\Delta^{k-1}$, i.e.~it is fully viable, every agent has viability 1.


The local viability of $a_i$ again counts the number of hyperedges that include $a_i$ as a vertex, but now normalized relative to the number of vertices in its neighborhood. Since there is no dependence on $k$, this provides a more local metric compared to viability. Low local viability indicates that an agent prefers or is only able to form small coalitions. For example, if an agent only forms coalitions with one other agent at a time (so its neighborhood looks like a graph), local viability is 0.
Higher local viability indicates an agent's willingness to form many coalitions with its adjacencies, demonstrating more flexibility. As an extreme, we have

\begin{prop}\label{P:MaxAdapt}
The local viability of an agent $a$ is 1 if and only if the agent forms all possible coalitions with its adjacencies. 
\end{prop}

\begin{proof}
Local viability of an agent $a$ is 1 if and only if $|N(a)|=2^{|\ac(a)|}$, and this is true if an only if $a$ forms a hyperedge with every subset of its adjancencies. 
\end{proof}

One way to think about the above is that, if local viability is 1, then $a$ and its adjacencies form a $\Delta^{|\ac(a)|}$, possibly with the face determined by the adjacencies alone (the face ``opposite'' $a$) and some or all of its subfaces missing. If any of the other faces are missing, $|N(a)|$ decreases and local viability is strictly less than 1. In the case when $H$ is a simplicial complex, this means that an agent has local viability 1 if they belong to exactly one maximal coalition.

We will make some further comments about local viability and its relation to viability after working out some examples.

The stability of $P$ is the  total number of hyperedges of $P$, again normalized. This is essentially the sum of the viabilities of all the agents, but taking into account the overcount of hyperedges. Greater $\operatorname{stab}(P)$ indicates more compatibilities among agents and more willingness to form coalitions. As one would expect, a fully viable system (simplex) has stability 1 while the system with no coalitions has viability 0.


\begin{example}\label{ex:localviability}
Consider political structures \(P_1, P_2, P_3, P_4\), each containing five agents \(a_1, a_2, a_3, a_4, a_5\). In each, the neighborhood \(N(a_1)\) of $a_1$ has 4 elements:

\begin{align*}
\text{In } P_1: & \ N(a_1) = \{\{a_1\}, \{a_1, a_2\}, \{a_1, a_3\}, \{a_1, a_2, a_3\}\} \\
\text{In } P_2: & \ N(a_1) = \{\{a_1\}, \{a_1, a_2\}, \{a_1, a_2, a_3\}, \{a_1, a_4\}\} \\
\text{In } P_3: & \ N(a_1) = \{\{a_1\}, \{a_1, a_2, a_3\}, \{a_1, a_4, a_5\}, \{a_1, a_2, a_3, a_4, a_5\}\} \\
\text{In } P_4: & \ N(a_1) = \{\{a_1\}, \{a_1, a_2\}, \{a_1, a_3\}, \{a_1, a_4, a_5\}\}
\end{align*}
%
%

The viability of \(a_1\) is the same in all four, since, in each structure,
\[
\operatorname{via}(a_1) = \frac{4-1}{{2^{5-1}-1}} = \frac{1}{5}.
\]

However, the local viability of \(a_1\) differs:

In \(P_1\)\ : \(\ac(a_1) = 2\), and \(\operatorname{lvia}(a_i)\) is 1. This makes sense in light of Proposition \ref{P:MaxAdapt} because the neighborhood of \(a_1\) consists of all possible hyperedges that contain one or both of $a_2$ and  $a_3$.

In \(P_2\)\ : \(\ac(a_1) = 3\), and \(\operatorname{lvia}(a_1)\) is \(3/7\).

In \(P_3\)\ : \(\ac(a_1) = 4\), and \(\operatorname{lvia}(a_1)\) is \(1/5\). Note that this value is the same as \(a_1\)'s viability. This is because $a_1$ is connected to all other agents so the local and global measures coincide.

In \(P_4\)\ : \(\ac(a_1) = 4\), and \(\operatorname{lvia}(a_1)\) is \(1/5\).


\end{example}


While local viability offers more insight than viability in this example, it still cannot distinguish between the cases of  \(P_3\) and \(P_4\). To further detect differences between \(P_3\) and \(P_4\), one might consider bringing in the $s$-vector of a vertex.

The following example illustrates that viability, with its global view of the position of an agent, can also provide a more nuanced view than local viability. 


\begin{example}\label{ex:localviability2}
Consider these political structures:

\begin{align*}
 P_1 & = \{\{a_1\}, \{a_2\}, \{a_3\}, \{a_1, a_3\}, \{a_1, a_2, a_3\}\}\\
P_2  & = \{\{a_1\}, \{a_2\}, \{a_3\}, \{a_4\}, \{a_5\}, \{a_1, a_3\}, \{a_1, a_2, a_3\}, \{a_3, a_4\}, \{a_3, a_5\}, \{a_4, a_5\}, \{a_2, a_3, a_5\}\}
\end{align*}

The local viability of \(a_1\) is equal in the two structures since the neighborhoods of \(a_1\) are the same. 

However, their viabilities are:

In \(P_1\): \(\operatorname{via}(a_1) = \frac{{3-1}}{{2^{3-1} - 1}} = \frac{2}{3}\).

In \(P_2\): \(\operatorname{via}(a_1) = \frac{{3-1}}{{2^{5-1} - 1}} = \frac{2}{15}\).

%
%

\end{example}

Small viability indicates a low overall level of interaction between an agent and the other agents. If local viability is simultaneously high, this agent forms many coalitions, but with a small subset of agents. This might point to the existance of a clique or a faction within a political structure with the agent in question as the ringleader.

%

\begin{example}
For a complete hypergraph $\Delta^{k-1}$, \( |\operatorname{N}(a_i)| = 2^{k-1} \) for all \( a_i \) so the viability of each agent is 1. The adjacency count of each agent is \( k-1 \), so the local viability of each agent is also 1. Since the number of hyperedges in a complete hypergraph is \( 2^k - 1 \), the stability  is also 1.

At the other extreme is the situation where the structure consists of $k$ isolated vertices so no agents are compatible. Now \( |\operatorname{deg}(a_i)| = 0 \) for all \( a_i \), and the viability of each agent is 0. For any agent, \( |\operatorname{N}(a_i)| = 1 \), as the hyperedge containing the vertex itself is the sole hyperedge in its neighborhood. Consequently, the local viability of each agent is 0. Since the only hyperedges are those consisting of a single vertex, \( |P| = |V(P)| = k \) and so the stability of \( P \) is 0.
\end{example}


\subsection{Operations on political structures}\label{S:PoliticsMods}


Hypergraph constructions from Section \ref{S:HypergraphOperations} carry over into the setting of political structures in a straightforward way. This section explains how and establishes some results about the interplay between the notions of the previous section and these constructions.

Suppose a political structure \( P \) with \( k\geq 2 \) agents consists of two distinct path-connected components, $P_1$ and $P_2$, so $P=P_1\amalg P_2$. Suppose an agent $a_0\in P_1$ and an agent $b_0\in P_2$ decide to merge or join forces, effectively becoming a single agent. If we rename the newly formed agent by $a$, we then have a situation that corresponds to the wedge of hypergraphs $P_1\vee P_2$ (see \refD{Wedge}). 

\begin{prop}\label{P:WedgeViability}
If agents are merged from two path-connected components of a political structure, the viability of the newly formed agent is greater than that of each agent individually. In fact, 
$$
\operatorname{via}(a)>\operatorname{via}(a_0)+\operatorname{via}(b_0).
$$
\end{prop}

\begin{proof} From the definition of viability, $|N(a_0)| = (2^{k-1} -1 )\cdot \operatorname{via}(a_0)  + 1$. Similarly for $b_0$. By definition of the wedge, \(|\operatorname{N}(a)| = |\operatorname{N}(a_0)| + |\operatorname{N}(b_0)| - 1\), where the subtraction of 1 is because the hyperedge containing only the wedge vertex itself should be  counted only once. Additionally, the number of agents decreases by \(1\) after the wedge operation.

Hence
\begin{align*}
\operatorname{via}(a) &= \frac{1}{2^{k-2}-1} \cdot (|\operatorname{N}(a_0)| + |\operatorname{N}(b_0)| - 2) \\
&= \frac{1}{2^{k-2}-1} \cdot \left((2^{k-1} -1 )\cdot \operatorname{via}(a_0) + 1 + (2^{k-1} -1 ) \cdot \operatorname{via}(b_0) + 1 - 2\right) \\
&= \frac{1}{2^{k-2}-1} \cdot \left((2^{k-1} -1 )\cdot (\operatorname{via}(a_0) + \operatorname{via}(b_0)) \right) \\
&= \frac{2^{k-1} - 1}{2^{k-2} - 1}\cdot (\operatorname{via}(a_0) + \operatorname{via}(b_0))
\end{align*}
Since \(\frac{2^{k-1} - 1}{2^{k-2} - 1} > 1\) when \( k > 2 \), it follows that \(\operatorname{via}(a) > \operatorname{via}(a_0) + \operatorname{via}(b_0)\).
\end{proof}

For local viability, the situation is different. The discrepacy comes from the fact that $k$ is the same for both $a_0$ and $b_0$ in the above, while the adjacency count for the two can be different.

\begin{prop}\label{P:WedgeLocalViability}
If agents are merged from two path-connected components of a political structure, the local viability satisfies
\[
\operatorname{lvia}(a) = \frac{2^{\ac(a_0)} - 1}{2^{\ac(a_0) + \ac(b_0)}-1} \cdot \operatorname{lvia}(a_0) + \frac{2^{\ac(b_0)} - 1}{2^{\ac(a_0) + \ac(b_0)}-1} \cdot \operatorname{lvia}(b_0).
\]
In particular, since the two fractions are less than 1 for $\ac(a_0), \ac(b_0)\geq 1$, we have
$$
\operatorname{lvia}(a) <\operatorname{lvia}(a_0) + \operatorname{lvia}(b_0). 
$$
\end{prop}

\begin{proof}
As in the above proof, we have $|\operatorname{N}(a_0)| = (2^{\ac(a_0)} - 1)\operatorname{lvia}(a_0) + 1$ and similarly for $b_0$. Since $\ac(a) = \ac(a_0) + \ac(b_0)$, 
\begin{align*}
\operatorname{lvia}(a) &= \frac{1}{2^{\ac(a_0) + \ac(b_0)}-1} \cdot \left( |\operatorname{N}(a_0)| + |\operatorname{N}(b_0)| - 2 \right) \\
&= \frac{1}{2^{\ac(a_0) + \ac(b_0)}-1} \cdot  \left( (2^{\ac(a_0)} - 1) \cdot \operatorname{lvia}(a_0) + 1  + (2^{\ac(b_0)} - 1) \cdot \operatorname{lvia}(b_0) + 1 - 2 \right) \\
&= \frac{2^{\ac(a_0)} - 1}{2^{\ac(a_0) + \ac(b_0)}-1} \cdot \operatorname{lvia}(a_0) 
 + \frac{2^{\ac(b_0)} - 1}{2^{\ac(a_0) + \ac(b_0)}-1} \cdot \operatorname{lvia}(b_0)
\end{align*}
\end{proof}

\begin{prop}\label{P:WedgeStabilityComponents}
Merging agents from different path-connected components of a political structure increases stability:
$$
\operatorname{stab}(P_1\amalg P_2)< \operatorname{stab}(P_1\vee P_2).
$$
\end{prop}

\begin{proof}
The stability before the merging of the agents is given by 
\[
\operatorname{stab}(P_1\amalg P_2) = \frac{1}{2^k - 1 - k} \cdot (|P_1\amalg P_2| - k)
\]
where by $|\cdot|$ we as usual mean the number of hyperedges of the structure.
After the merging, there is one fewer agent and hence one fewer hyperedge. Therefore, the stability after the merging is
\[
\operatorname{stab}(P_1\vee P_2) = \frac{1}{2^{k-1} - 1 - (k-1)} \cdot (|P_1\amalg P_2| - 1 - (k-1))
\]
or
\[
\operatorname{stab}(P_1\vee P_2) = \frac{1}{2^{k-1} - k} \cdot (|P_1\amalg P_2| - k)
\] 

Since \(2^k - 1 - k > 2^{k-1} - k\) when \(k > 2\), 
\(\frac{1}{2^k - 1 - k} < \frac{1}{2^{k-1} - k}\). 
Therefore, \(\operatorname{stab}(P_1 \amalg P_2) < \operatorname{stab}(P_1 \vee P_2)\).
\end{proof}

Two political structures can also merge along two matching substructures (subhypergraphs). This is known as the \emph{pushout} and it would be interesting to generalize Proposition \ref{P:WedgeStabilityComponents} to that situation.

A similar yet distinct situation that can also be modeled using the wedge model is that, instead of being two components of one political structure, \(P_1\) and \(P_2\) are initially two separate political structures and are then merged along two agents. There seems to be no easy relationship between the stability of the individual structures and the merged one. 

For example, if 
\begin{align*}
P_1 & =\{ \{v_1\}, \{v_2\}, \{v_3\}, \{v_1, v_2\}, \{v_2, v_3\}, \{v_1, v_2, v_3\} \}\\
P_2 & =\{ \{w_1\}, \{w_2\}, \{w_3\}, \{w_1, w_3\}, \{w_1, w_2, w_3\} \}
\end{align*}
then merging along $v_1$ and $w_1$ produces a structure whose stability is lower than that of $P_1$ and $P_2$ individually. On the other hand, if $P_1$ consists of ten agents and a single viable configuration containing all of them while $P_2$ is a 5-simplex, then the merged structure (along any two vertices) has stability greater than $P_1$ but less than $P_2$.

Sometimes we want to introduce a \emph{mediator} into a political structure $P$. This is a new agent who is willing to negotiate or enter coalitions with any existing agents of viable configurations. This is modeled by $CP$, the cone on $P$. It is easy to show that the stability in this case increases, and in fact the proof of the same result for simplicial complexes, \cite[Proposition 3.14]{MV:Politics}, applies here.

The question of introducing a mediator into a substructure of $P$ is less straightforward (and more interesting). Such a scenario could arise, for instance, when certain agents are deeply divided, necessitating separate mediation to facilitate their participation in broader discussions. This situation can be modeled by taking a cone on a subhypergraph of \( P \). In this case, the stability of the political structure does not necessarily increase, which may seem counterintuitive. The weakening occurs as the mediator strenghtens the viable configuration structure among some agents but may have the effect of isolating them from the rest of the system. For details, see Proposition 3.15 in \cite{MV:Politics} whose proof applies when simplicial complexes are replaced by hypergraphs.

\begin{rem}
In \cite{MV:Politics}, an extension to \emph{weighted} political structures is considered. This is the model for the scenario where agents are deciding their position on more than one issue and each simplex is labeled by how many issues the agents forming that coalition agree on. Weighted versions of viability and stability for hypergraphs can be defined the same way they are for simplicial complexes and all the results about weighted political structures from \cite{MV:Politics} carry over to the setting of hypergraphs in the same way.
\end{rem}


\subsection{Homology and political structures}\label{S:StructuresHomology}


In this section, we initiate a study of how embedded homology can be used to extract information about the dynamics in a political structure. One observation is that all the results of \cite[Section 3.5]{MV:Politics} hold in the case of hypergraphs and embedded homology since this homology defines cycles the same way as standard homology (the difference is that there are fewer of them in embedded homology). Thus in particular, a non-zero Betti number (rank of a homology group) points to non-viabilities among certain subsets of agents.

However,  it is also helpful to play embedded and standard homology off each other. Neither is robust enough to be sufficiently useful on its own; homology is homotopy invariant, and in much of what we do, we want to keep track of structures that are not preserved by homotopy equivalences. For example, if we know $\operatorname H^{emb}_0(H) = \mathbb{Z} \oplus \mathbb{Z}$, this means that the agents are partitioned into two subsets with no hyperedges of size 2 between them. If there are $k$ agents, we can say that $H$ is missing at least $k-1$ hyperedges from a complete hypergraphs, but that is about all the information we can extract.

But if we know both embedded and singular homology of a hypergraph, we can extract more information about the possible viable configurations in the structure. For example,  for the left image in Figure \ref{fig:closure}, $\operatorname H^{emb}_0(H) = \mathbb{Z} \oplus \mathbb{Z}$. However, this homology would be the same even if the viable configuration $\{v_1,v_2,v_3\}$ was not there. On the other hand, the standard homology of the closure of that figure (the right side of Figure \ref{fig:closure}) is $\operatorname H_0(K_H) = \mathbb{Z}$. But if the hyperedge $\{v_1,v_2,v_3\}$ was not there, then $\operatorname H_0(K_H) = \mathbb{Z}\oplus \mathbb{Z}$.  In fact, it is not hard to see that zeroth and first embedded and singular homology groups together completely determine the political structure on three agents. 

Here is a sample result of what we believe could be a fruitful line of investigation. The statement is about the 0th homology (i.e.~path-connected components), but a generalization to higher homology groups is also possible. Let $\beta^{emb}_0(H)$ and $\beta_0(K_H)$ be the 0th Betti numbers, i.e.~ranks of $\operatorname H^{emb}_0(H)$ and  $\operatorname H_0(K_H)$, respectively. 

\begin{prop} For any hypergraph $H$,
\begin{itemize}
\item $\beta^{emb}_0(H)\geq \beta_0(K_H)$. 
\item If the above is an equality, then the structure is path connected, i.e.~there is a path through hyperedges of cardinality 2 from any agent to any other agent.
\item  If the inequality is strict, then there are hyperedges of size greater than two whose union contains vertices from $\beta^{emb}_0(H)- \beta_0(K_H)$ different components of the hypergraph.


\end{itemize}

%
%

\end{prop}

\begin{proof}
Zeroth homology in general count the number of path-connected components of a space, i.e.~it keeps track of paths of edges of cardinality 2. Taking the closure of a hypergraph might introduce such edges as faces of other hyperedges. Thus $K_H$ might contain more paths than $H$ itself and introduce edges of size 2 between distinct path-components of $H$, thereby decreasing their count. The only way this can happen is if there are hyperedges of size 3 or greater containing vertices from different components. These hyperedges provide new paths once the closure is taken.
\end{proof}

If the difference in Betti numbers is small, this indicates a fragmentation in the structure with lots of clusters of agents but few coalitions between the clusters. If the number is high, this means that the coalitions between clusters exist, but they are potentially large. Note the extreme case when there is a single large coalition containing all agents. Then $\beta_0(K_H)=1$.

\subsection{Delegating power}\label{S:Delegations}


%
%
%
%

In political structures, agents often delegate or cede authority to other, more favorably positioned agents. In \cite{AK:Conflict}, the better-posioned agent is said to be \emph{more central} and is modeled in the simplicial complex setting by a more dominant vertex. An elementary strong collapse is then precisely the the model for relinquishing authority. This is called a \emph{friendly delegation} \cite{AK:Conflict}. It was shown in \cite{MV:Politics} that the notion of viability and the process of delegating power interface in a compatible way, namely the viability of an agent increases if another agent delegates power to them, as one would expect.  

We extend this result below to the setting of hypergraphs. We could quote the proof of \cite[Theorem 4.1(1)]{MV:Politics} and say that it goes through the same way here, but instead we offer a more streamlined version.



\begin{prop} Suppose $P$ is a political structure with $k$ agents and agent $b$ is dominated by agent $a$. Relabel agent $a$ as $a'$ after a friendly delegation from  \( b \) to \( a \). Then
$$
\operatorname{via}(a') \geq \operatorname{via}(a).
$$
\end{prop}

\begin{proof}
The viability of agent \( a \) is initially
\[
\operatorname{via}(a) = \frac{|\operatorname{N}(a)| - 1}{2^{k-1} - 1}.
\]
After the friendly delegation from \( b \) to \( a \), the number of agents decreases by 1. There is also a decrease in \( |\operatorname{N}(a)| \); denote the difference by  \( n \), so that
\[
\operatorname{via}(a') = \frac{|\operatorname{N}(a)| - n - 1}{2^{k-2} - 1}.
\] 

Since \(n\) is the number of edges involving both \(a\) and \(b\) before the friendly delegation, \(n \leq 2^{k-2}\). We also have \(|\operatorname{N}(a)| \leq 2^{k-1}\). Using these two inequalities, it is easy to deduce that \(\operatorname{via}(a') \geq \operatorname{via}(a)\). 
\end{proof}

Equality is achieved when \(|\operatorname{N}(a)| = 2^{k-1}\) and \(n = 2^{k-2}\), resulting in \(\operatorname{via}(a) = 1\). I.e.~equality only occurs when $a$ is fully viable to begin with.

The situation with local viability is not as straighforward. For example, consider the political structure
\begin{align*}
P = \{ &\{a\}, \{b\}, \{c\}, \{d\}, \{a, b\}, \{b, c\}, \{b, d\},\{c, d\},
       \{b, c, d\}, \{a, b, c, d\} \}.
\end{align*}
Here $a$ is more central than $b$, but after $b$ delegates to $a$, the resulting structure is $\{\{a\}, \{b,c\}  \}$. Now $a$ is an isolated agent with its local viability going down to 0 from $2/7$. But if we amend $P$ by adding in the configuration $\{ \{a,c\}  \}$ and then delete $b$, producing the structure $\{\{a\}, \{a,c\}, \{b,c\}  \}$, $a$'s viability increases from $3/7$ to 1.

The difference between these two situations is that, in the first, agent $a$ had no connections to agents $c$ or $d$, so when $b$ delegated power to $a$, both $a$ and the pair $\{ c, d\}$ became isolated. In a sense, the structure of the system was damaged because of $b$'s poor judgement to delegate to $a$. This does not happen in the second example since $a$ and $c$ form a coalition prior to delegation. Another measure of this is that the stability of $P$ in both examples decreases, but the decrease is smaller in the second situation.

Note that, in the first example, $a$ is also dominated by $b$, so that $a$ could also delegate to $b$. This would increase the local viability of $b$, indicating that this is a better move since the coalition structure would be stronger after this delegation. This provides a potential delegation strategy -- delegation should only occur if the local viability of the agent being delegated to increases.

%
%
%
%
%
%
%
%
%
%


\section{Future work}\label{S:Future}


This paper sets up the basic framework of modeling political structures by hypergraphs and there are still many potential directions of further investigation. For example, one could take any of the standard notions and constructions from the theory of hypergraphs (and simplicial complexes) and try to translate them into this setting. One such is the \emph{clustering coefficient} \cite{ER:NetworksHypergraphs, HS:Clustering} which, in its various forms, can be defined to measure both local hypergraph structure or provide aggregated infromation about the entire hypergraph (some even result in approximations as necessitated by computational complexity). Our notion of local viability could in particular be upgraded to some kind of a clustering coefficient which might be a more powerful tool of gauging the power of an agent. Especially interesting would be a connection to the game-theoretic approach to hypergraph clustering \cite{RP:ClusteringGame} since the work in this paper borders game theory and social choice theory.

We would also like to extend the Banzhaf and Shapley-Shubik power indices from simplicial complexes, as was done in \cite{SV:SimplicialPower}, to hypergraphs. This would supply a refined notion of power distribution in weighted voting systems where certain coalitions are unfeasible or forbidden. In addition, each hyperedge is geometrically a simplex emebedded in some Euclidean space, so one could interpret the position in the simplex as influence of each agent in the corresponding coalition. This information could be added to the calculation of power indices. 

Our elementary foray into (embedded) homology deserves to be explored further. Higher homology groups should be incorporated into the picture.  The existence of a homology cycle indicates incompatibilities among a subset of agents. While this was relatively easy to describe in the case of simplicial complexes, hypergraphs are a different story because of the many ways subsets of agents forming a cycle may or may not form (sub)coalitions. What is likely required is a consideration of all homology groups at the same time, suggesting a combinatorial difficulty but also a structural richness and potential for a powerful tool.

We have not examined maps between hypergraphs, but they could play an important role. A hypergraph map sends hyperedges to hyperedges (and faces to faces) and as such captures a change in a political system that causes agents or substructures to be consolidated or introduces new ones. The collection of hypergraphs with hypergraph maps forms a category, which means that political structures inherit a category structure as well. This could provide a useful framework and an organizing mechanism to study political structures. Hypergraphs are are category-theoretically related to simple games from social choice theory and political structures can also easily be brought into that framework.



\def\cprime{$'$}
\providecommand{\bysame}{\leavevmode\hbox to3em{\hrulefill}\thinspace}
\providecommand{\MR}{\relax\ifhmode\unskip\space\fi MR }
\providecommand{\MRhref}[2]{%
  \href{http://www.ams.org/mathscinet-getitem?mr=#1}{#2}
}
\providecommand{\href}[2]{#2}

\end{document}